\documentclass[paper=a4, pagesize, fontsize=11pt, DIV=11, headings=small, parskip=half]{scrartcl}

\usepackage[utf8]{inputenc}
\usepackage[T1]{fontenc}
\usepackage{amsmath,amssymb,commath}
\usepackage{aliascnt,cleveref,bbm,enumerate,framed,graphicx,natbib,xcolor,etoolbox}
\usepackage[thmmarks,framed,amsmath]{ntheorem}

\input{macros.pub}

\title{\Large Directed Random Walks on Polytopes with Few Facets}
\author{\parbox{\textwidth}{\footnotesize\centering Malte Milatz\\ mmilatz@inf.ethz.ch\\ Department of Computer Science, ETH Zürich}}
\date{}

\begin{document}

    \maketitle
    \begin{abstract}
        \vspace{-2em}
        \noindent
        Let $P$ be a simple polytope with $n-d = 2$, where $d$ is the dimension
        and $n$ is the number of facets.
        The graph of such a polytope is also called a \emph{grid}.
        It is known that the directed random walk along the edges of $P$ terminates
        after $O(\log^2 n)$ steps, if the edges are oriented in a (pseudo-)linear
        fashion.
        We prove that the same bound holds for the more general unique sink
        orientations.
    \end{abstract}

\section{Introduction}

Our research is motivated by the simplex algorithm for linear programming.
We consider the variation where the algorithm chooses at each step the next
position uniformly at random from all improving neighbouring positions;
this rule is commonly called \textsc{Random-Edge}.
Its expected runtime on general linear programs can be mildly exponential;
cf.~\cite{FriedmannHZ'11}.
Better bounds can be hoped for if one imposes restrictions on the input.
It is intuitively plausible that \textsc{Random-Edge} should run very fast if the
number of constraints (or facets) is very small in relation to the dimension.

\cite{GaertnerSTVW'01} analyzed the performance of \textsc{Random-Edge}
on simple polytopes with $n$ facets in dimension $d = n-2$, and obtained the
tight bound $O(\log^2 n)$.
It is natural to ask inhowfar this bound depends on the geometry of the
problem.
To this end we consider the setting where the notion of `improving' is
specified not by a linear objective function, but by a \emph{unique sink
orientation}, which is a more general object with a simple combinatorial
definition.

Unique sink orientations have been studied in numerous contexts; see
e.g.~\cite{SzaboW'01,GaertnerS'06}, and \cite{GaertnerMR'08}.
They are defined as follows.
A \emph{sink} in a directed graph is a vertex without any outgoing edges.
Now, an orientation of the edges of a polytope is a \emph{unique sink
orientation} if every non-empty face of the polytope has a unique sink.
The definition is motivated by the fact that every linear orientation (the
orientation obtained from a generic linear objective function)
is a unique sink orientation (but the converse does not hold).

\begin{figure}
    \begin{center}
        \includegraphics[width=10em]{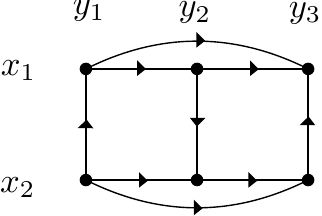}
    \end{center}
    \caption{The graph of a polytope (a prism)
        with $n=5$ facets and dimension $d=3$,
        which is a grid.
        Every vertex is identified by a pair $\{x_i,y_j\}$.
        The arrows give an example of a unique sink orientation.}
    \label{Figure: grid}
\end{figure}

The purpose of this note is to prove the following theorem.

\begin{theorem}
    \label{Theorem: upper bound}
    Let $n-d=2$.
    Let $P$ be a simple $d$-dimensional polytope with $n$ facets,
    endowed with a unique sink orientation.
    A directed random walk on $P$, starting at an arbitrary vertex,
    arrives at the sink after an expected number of $O(\log^2 n)$ steps.
\end{theorem}

\vspace{-1em}
To be perfectly clear, by a \emph{directed random walk} on a given directed
graph we mean the following process: We begin in some given vertex $v_0$; we
choose one edge uniformly at random from the set of outgoing edges at $v_0$;
we move to the other endpoint and call it $v_1$; then we continue in the same
fashion until we possibly arrive at a sink.

For the more restrictive \emph{Holt-Klee} (or \emph{pseudo-linear})
orientations, the bound in \cref{Theorem: upper bound} has been proved by
\cite{Tschirschnitz'03}.
The general structure of our proof is very similar to Tschirschnitz'; the only
notable deviation will be the proof of \cref{Lemma: hitting the next
milestone}.


In order to fix some notation, let $H$ be the set of the $n$ halfspaces that
define the polytope $P$, and let $V$ be the set of vertices.
We will identify every vertex $v \in V$ with the set
\[
    \{ h \in H \,:\, v \text{ lies in the interior of } h \},
\]
so that $V$ becomes a subset of the powerset of $H$.

The following fact appeared as lemma 2.1 in \cite{FelsnerGT'05}, and may
also be seen as a consequence of the existence of Gale diagrams:
Assuming that $P$ is a simple polytope with $n-d=2$,
there exists a partition $H = X \dot\cup Y$ such that,
under the identification of vertices with subsets of $H$ described
above, we have
\[
    V = \left\{ \{x,y\} ~:~ x \in X,~ y \in Y \right\}.
\]
Practically speaking, we can thus refer to each vertex by its $X$-coordinate
and $Y$-coordinate.
Furthermore, a set of vertices forms a face if and only if it is of the form
$V \cap 2^{H'}$ for some $H' \subseteq H$.
In particular, two vertices $v,v'$ are adjacent if and only if the sets
$v,v'$ are not disjoint.

The graph of the polytope $P$ is thus isomorphic to the product of two
complete graphs, as in \cref{Figure: grid}.
Such a graph is also called a \emph{grid}.
Those readers who are not used to unique sink orientations might at this point
want to check that the orientation shown in the figure is indeed a unique sink
orientation.
Note in particular that every row or column in fig.~\ref{Figure: grid} is also
a face of $P$, the edges of such a face constitute a complete graph, and they
must be oriented in an acyclic fashion.

We write $u \to v$ for a directed edge from a vertex $u$ to a vertex $v$.
A non-empty directed path from $u$ to $v$ is denoted by $u \to^+ v$.
The \emph{outmap} $\Phi : V \to 2^H$ specifies the outgoing edges at each
vertex and is defined by
    \begin{align*}
        \Phi(v) :=
        \bigl\{ h \in H \setminus v \,:\,
            v \to w \text{ for some } w \in V \text{ with } h \in w
        \bigr\}.
    \end{align*}
We also abbreviate
\[
    \Phi_X(v) := \Phi(v) \cap X \text{ and } \Phi_Y(v) := \Phi(v) \cap Y.
\]
For example, for the top middle vertex of the grid pictured in \cref{Figure:
grid}, we have $\Phi_X(\{x_1,y_2\}) = \{x_2\}$ and
$\Phi_Y(\{x_1,y_2\}) = \{y_3\}$.

\section{Milestones}

Given a vertex $v$, the pair $(|\Phi_X(v)|,|\Phi_Y(v)|)$ is known as the
\emph{refined out-degree}.
From \cite{FelsnerGT'05}, Lemma 3.1, we know that for every pair of indices
$(i,j)$ with $0 \le i < |X| - 1$ and $0 \le j < |Y| - 1$ there exists a unique
vertex with refined out-degree $(i,j)$.
We use this property to define the following `milestones' for our random walk.

The number of milestones will be
\[ L := 1 + \left\lfloor \log_2 \del{ \min \cbr{ |X|-1, |Y|-1 }  } \right\rfloor . \]
For all $i \in \{ 1 ~\dots~ L \}$, let $w^i$ denote the unique vertex of $P$ with
refined out-degree $(2^{i-1},2^{i-1})$.
(The indices are chosen in such a way that the vertex $w^i$ exists and
has exactly $2^i$ outgoing edges.)
Furthermore we define $w^0$ as the unique sink of $P$; in other words, $w^0$
is the unique vertex with refined out-degree $(0,0)$.
Now define $W^i$ as the set of vertices to which there is a
non-empty directed path from this vertex, i.\,e.,
\[ W^i := \{\, v \in V \,:\, w^i \to^+ v \,\}. \]
The sets $W^i$ serve as a measure of progress:
Starting from a vertex in the set $W^{i+1}$,
the next `milestone' is hit when the random walk arrives for the first time in
a vertex of $W^i$;
and once the random walk arrives in some $W^i$, it stays therein.
Note that the indices are counting down:
The random walk arrives in the global sink as soon as the milestone
$W^0 = \{w^0\}$ is hit.
Our goal is now to prove the following propositions.
\begin{proposition}
    \label{Proposition: next milestone}
    The expected time until the random walk,
    starting from a vertex in $W^{i+1}$, arrives in $W^i$,
    is bounded by $O(\log n)$. ($i = 0,\dots,L-1$.)
\end{proposition}
\begin{proposition}
    \label{Proposition: first milestone}
    The expected time until the random walk,
    starting from an arbitrary position, arrives in $W^L$,
    is also bounded by $O(\log n)$.
\end{proposition}

\vspace{-2ex}
The bound $O(\log^2 n)$ in \cref{Theorem: upper bound} will follow
by observing that there are $L+1 = O(\log n)$ many
milestones, each of which is hit -- according to the two propositions -- after
at most $O(\log n)$ steps in expectation.

The technical statements of the following lemma will be useful for deducing
the orientation of some edges incident to the current position.
The statements are known or follow readily from known results.\footnote{In
particular, parts (a) and (b) of the lemma appear as lemma 4.6 in
\cite{Tschirschnitz'03}; part (c) is lemma 4.8 in \cite{Tschirschnitz'03};
part (d) can be proved using (c); and part (e) is from \cite{GaertnerMR'08}.}

\begin{lemma}
    \reallynopagebreak
    \label{Lemma: list}
\begin{enumerate}[(a)]
    \item 
        \label{Item: constraint containment}
        If a vertex $v$ satisfies $v \subseteq \Phi(w^i)$, then $v \in W^i$.
    \item
        \label{Item: constraint containment 2}
        For all $v \in W^i$, $v \cap \Phi(w^i) \neq \emptyset$.
    \item
        \label{Item: pivot containment}
        Let $v,w \in V$ and assume $w \not\to^+ v$.
        Then either $\Phi_X(v) \supseteq \Phi_X(w)$ or
        $\Phi_Y(v) \supseteq \Phi_Y(w)$ holds.
    \item
        \label{Item: union 5}
        $
            |\Phi(w^{i+1}) \cup \Phi(w^i)| \,\le\, 5 \cdot 2^{i-1}.
        $
    \item
        \label{Item: acyclic}
        Every unique sink orientation of $P$ is acyclic.
\end{enumerate}
\end{lemma}

\section{Proof of \cref{Proposition: next milestone}}
\label{Section: proof}

We write $v_0,v_1,\dots$ for the positions of the random walk,
where we consider the starting position $v_0$ as a fixed element of $V$
(not a random variable).
Let $i$ be chosen such that $v_0 \in W^{i+1}$, and let $T$ denote the hitting
time of the set $W^i$:
\[
    T \,=\, \min \bigl\{ k ~:~ v_k \in W^i \bigr\}.
\]
We want to bound the expected time until the random walk arrives in $W^i$; in
other words, we want to bound $\E T$.

One way to look at the directed random walk is as follows:
At time $k+1$ it picks a \emph{pivot} $h_{k+1}$ uniformly at random from the
set $\Phi(v_k)$.
This pivot determines the edge along which to move away from the vertex $v_k$.
Concretely, the next position $v_{k+1}$ is the unique neighbour of $v_k$ that
satisfies $v_{k+1} \subseteq v_k \cup \{ h_{k+1} \}$.
Note that the pivot $h_{k+1}$ is only defined in this way when the position
$v_k$ is not already the global sink; so if $v_k = w_0$ is the global sink
then we let $h_{k+1} = \Diamond$, where $\Diamond$ is just a formal symbol
to remind us that the random walk has already terminated.

We define some auxiliary stopping times.
Let $\sigma$ denote the first time that an element of the set $\Phi(w^{i+1})$
is pivoted.
Furthermore let $\tau_1 < \tau_2 < \dots$ be the instants in time when an
element of the set $\Phi(w^{i+1}) \cup \Phi(w^i)$ is pivoted,
and let $\tau_N$ be the first among these instants when the random walk
hits the set $W^i$.

More precisely, we let
\begin{align*}
    \sigma ~&:=~ \min \bigl\{\, k > 0 ~:~ h_k \in \Phi(w^{i+1}) \text{ or }
    h_k = \Diamond \,\bigr\},
    \\
    \tau_0 ~&:=~ 0,
    \\
    \tau_{j} ~&:=~
        \min \bigl\{\, k > \tau_{j-1} ~:~ h_k \in \Phi(w^{i+1}) \cup \Phi(w^i)
        \text{ or } h_k = \Diamond \,\bigr\}
        \qquad
         (j \ge 1),
    \\
    N ~&:=~ \min \bigl\{\, j ~:~ v_{\tau_j} \in W^i \,\bigr\}.
\end{align*}
We suppress the dependence on $i$ in the notation, considering $i$ (the index
of the next milestone) fixed throughout the section.

\begin{lemma}
    \label{Lemma: adjacency}
    The random set $\{ h_1, \dots, h_{\sigma - 1} \}$ is always either a
    subset of $X$ or a subset of $Y$.
    As a consequence, the vertices $v_0, \dots, v_{\sigma - 1}$ share either
    their $Y$-coordinate or their $X$-coordinate.
\end{lemma}

\begin{proof}
    Assume we encounter the event that, say, $h_2 \in X$ and $h_3 \in Y$,
    where $3 < \sigma$.
    Then $v_3 = \{ h_2, h_3 \}$.
    By definition of our stopping time $\sigma$, none of our pivots
    considered here are elements of $\Phi(w^{i+1})$; hence
    we have $v_3 \cap \Phi(w^{i+1}) = \emptyset$.
    On the other hand, by our choice of $i$ we have
    $v_3 \in W^{i+1}$
    and thus $v_3 \cap \Phi(w^{i+1}) \neq \emptyset$ by
    \cref{Lemma: list}(\ref{Item: constraint containment 2}); a contradiction.
\end{proof}

\begin{lemma} 
    \label{Lemma: log n}
    We have
    $
        \E \bigl[ \sigma \bigr] \,\le\, \Harm_n + 1
    $,
    where $\Harm_n$ denotes the $n$th harmonic number.
\end{lemma}

\begin{proof}
    By \cref{Lemma: adjacency}, either the pivots $h_0, \dots, h_{\sigma-1}$
    are all elements of $X$, or they are all elements of $Y$.
    Thus it suffices to bound the expected time until, say, the pivot is not
    an element of $Y$.
    In terms of \cref{Figure: grid}, this means to bound the expected time
    until the random walk leaves the current row of the grid.
    This expectation only becomes larger if we condition on the event that the
    random walk stays in the current row until it reaches the sink of the row.
    Reaching the sink of the row can be shown to take $\Harm_n$ steps in expectation;
    we still have to add $1$ for the possible additional step that leaves the
    current row.
\end{proof}
    

\begin{lemma}[Hitting the next milestone]
    \label{Lemma: hitting the next milestone}
    We have
    $
    \E T \,\le\, 155 (\Harm_n + 1).
    $
\end{lemma}

\begin{proof}
    We have $T \,\le\, \tau_N$; so we will concentrate our efforts on
    bounding the expectation of $\tau_N$.
    In the following we will need to make the starting position $v_0$ of the
    random walk explicit in the notation; we will do so by writing the
    starting position as a subscript, as in $\E_{v_0}[\_]$ or $\Pr_{v_0}[\_]$.

    Using the Markov property we find, for all $j \ge 1$,
    \begin{align*}
        \E_{v_0} \bigl[ \tau_j - \tau_{j-1} \given j \le N \bigr]
        \,&=\,
        \E_{v_0} \bigl[ \tau_j - \tau_{j-1} \given v_{\tau_{j-1}} \notin W^i \bigr]
        \\
        \,&=\,
        \sum\nolimits_{u \in W^{i+1} \setminus W^i } \Pr \bigl[ v_{\tau_{j-1}} = u \bigr]
        \E_{v_0} \bigl[ \tau_j - \tau_{j-1} \given v_{\tau_{j-1}} = u \bigr]
        \\
        \,&=\,
        \sum\nolimits_{u \in W^{i+1} \setminus W^i } \Pr\nolimits_{v_0} \bigl[ v_{\tau_{j-1}} = u \bigr]
        \E_u \bigl[ \tau_1 \bigr]
        \\
        \,&\le\,
        \sup\nolimits_{u \in W^{i+1}} \E_u \bigl[ \tau_1 \bigr]
        \\
        \,&\le\,
        \sup\nolimits_{u \in W^{i+1}} \E_u \bigl[ \sigma \bigr]
        \\
        \,&\le\,
        \Harm_n+1,
    \end{align*}
    where the last step used \cref{Lemma: log n}.
    Hence,
    applying \cref{Lemma: multiplicative bound} (appendix) to the sequence
    $(\tau_j - \tau_{j-1})_{1 \le j \le N}$,
    \begin{align*}
        \textstyle
        \E T
        \,\le\,
        \E \bigl[ \tau_N \bigr]
         \,=\,
        \E\bigl[\, \sum_{j=1}^N (\tau_j - \tau_{j-1}) \,\bigr]
         \,\le\,
        \E N \cdot (\Harm_n + 1).
    \end{align*}
    It remains to show that the number $\E N$ can
    be bounded from above by the constant $155$.
    To this end we consider the following events.
    \begin{align*}
        \Ev_1 ~&:~ v_{\tau_1} \in W^i \text{ or } h_{\tau_1} \in \Phi(w^i);
        \\
        \Ev_2 ~&:~ v_{\tau_2} \in W^i \text{ or  } h_{\tau_2} \in \Phi_\xi(w^i), \text{
        where } \xi \in \{X,Y\} \text{ such that } h_{\tau_1} \notin \xi;
        \\
        \Ev_3 ~&:~ v_{\tau_3} \in W^i.
    \end{align*}

    \textbf{Claims.}
    For any choice of starting position we have
            \reallynopagebreak
    \begin{enumerate}[\itshape (i)]
        \item
            $\Pr \bigl[ \Ev_1 \bigr] \ge \frac{1}{5}$.
            \reallynopagebreak

        \item
            $\Pr \bigl[ \Ev_2 \given \Ev_1 \bigr] \ge \frac{1}{5}$.
            \reallynopagebreak

        \item
            \reallynopagebreak
            $\Pr \bigl[ \Ev_3 \given \Ev_1, \Ev_2 \bigr] \ge \frac{1}{5}$.
    \end{enumerate}

    \vspace{1ex}
    Note that the event $\Ev_3$ is equivalent to the event $N \le 3$
    (or in other words, the event that the next milestone is hit no later than
    at time $\tau_3$).
    Thus, once these claims are proved, we can conclude
    that $\E N$ can be bounded from above by the expected number of steps
    that it takes a Bernoulli process with parameter $p = \frac{1}{5}$ to hit
    3 successive successes.
    By \cref{Theorem: Bernoulli process} in the appendix this number is, as
    desired,
    \[ \frac{1 - p^3}{p^3(1-p)} = 155. \]

    \textit{Proofs of the claims.}
    \reallynopagebreak
    \begin{enumerate}[\itshape (i)]
        \item
            In order to show \emph{(i)}, it suffices to show
            $\Pr \bigl[ \Ev_1 \given v_{\tau_1} \notin W^i \bigr] \ge \frac{1}{5}$.
            Let us thus assume $v_{\tau_1} \notin W^i$, which means that
            the random walk has not yet hit the next milestone at time $\tau_1$.
            In particular, the random walk has not yet terminated at time
            $\tau_1$.
            We want to show that now the event $h_{\tau_1} \in \Phi(w^i)$ happens
            with probability at least $\frac{1}{5}$.

            By definition of the random walk, the $k$th pivot $h_k$ is chosen
            uniformly at random from the set of violating constraints,
            $\Phi(v_{k - 1})$.
            This is true for any time $k$ at which the random walk has not yet
            terminated; now we consider $k=\tau_1$:
            By definition of $\tau_1$, the pivot $h_{\tau_1}$ is then
            chosen (still uniformly at random) only from the smaller set
            \begin{align*}
                S \,:=\, \Phi(v_{\tau_1 - 1}) \cap (\Phi(w^{i+1}) \cup \Phi(w^i)).
            \end{align*}
            Since the next milestone has not yet been hit, we know that
            $v_{\tau_1 - 1} \notin W^i$ holds, which is equivalent to writing
            $w^i \not\to^+ v_{\tau_1 - 1}$.
            Thus \cref{Lemma: list}(\ref{Item: pivot containment}) implies that
            either $\Phi_X(v_{\tau_1-1}) \supseteq \Phi_X(w^i)$ or
            $\Phi_Y(v_{\tau_1-1}) \supseteq \Phi_Y(w^i)$ holds.
            In both these cases we see that $S$ contains
            a subset of $\Phi(w^i)$ of size $2^{i-1}$.
            On the other hand,
            \[ |S| \le |\Phi(w^{i+1}) \cup \Phi(w^i)| \le 5 \cdot 2^{i-1} \]
            by \cref{Lemma: list}(\ref{Item: union 5}).
            Hence,
            \begin{align*}
                \Pr \bigl[ \Ev_1 \given v_{\tau_1} \notin W^i \bigr]
                \,=\,
                \Pr \bigl[ h_{\tau_1} \in \Phi(w^i) \given
                v_{\tau_1} \notin W^i \bigr] \,&\ge\, 2^{i-1} / (5 \cdot
                2^{i-1}) \,=\, \frac{1}{5}.
            \end{align*}

        \item
            Similarly to how we proceeded in the proof of \emph{(i)},
            it suffices to show $\Pr\bigl[ \Ev_2 \given \Ev_1 \text{ and }
            v_{\tau_2} \notin W^i \bigr] \ge \frac{1}{5}$.
            So we assume $\Ev_1$ and $v_{\tau_2} \notin W^i$, and without loss
            of generality we assume $h_{\tau_1} \in \Phi_X(w^i)$.
            We want to show that now $h_{\tau_2} \in \Phi_Y(w^i)$ happens with
            probability at least $\frac{1}{5}$.

            This time the pivot $h_{\tau_2}$ is chosen uniformly at random
            from the set
            \[
                S' \,:=\, \Phi(v_{\tau_2-1}) \cap (\Phi(w^{i+1}) \cup \Phi(w^i)).
            \]
            As before, $|S'| \le 5 \cdot 2^{i-1}$.
            Also as before, either $S' \supseteq \Phi_X(w^i)$ or $S' \supseteq
            \Phi_Y(w^i)$ must hold.
            Now, however, we can observe that the latter alternative must be
            true (implying the claim because $|\Phi_Y(w^i)| = 2^{i-1}$), as
            follows:

            It suffices to show that $h_{\tau_1} \notin \Phi(v_{\tau_2 - 1})$.
            (Indeed, this implies $S' \not\supseteq \Phi_X(w^i)$, which leaves us
            only with the other alternative, $S' \supseteq \Phi_Y(w^i)$.)
            We know from \cref{Lemma: adjacency} that the vertices
            $v_{\tau_1}$ and $v_{\tau_2 - 1}$ share either their
            $X$-coordinate or their $Y$-coordinate.
            If they share their $X$-coordinate, so that
            $h_{\tau_1} \in v_{\tau_2 - 1}$,
            then clearly $h_{\tau_1} \notin \Phi(v_{\tau_2 - 1})$.
            If on the other hand our two vertices share their $Y$-coordinate
            but not their $X$-coordinate,
            let us consider the grid edge between our two vertices:
            Since there is a walk from $v_{\tau_1}$ to $v_{\tau_2 - 1}$, the
            edge cannot be directed from $v_{\tau_2 - 1}$ to $v_{\tau_1}$,
            by acyclicity (\cref{Lemma: list}(\ref{Item: acyclic})).
            Note furthermore that we have $v_{\tau_1} = \{ h_{\tau_1}, y \}$,
            where $y$ is the shared $Y$-coordinate.
            Hence the non-existence of a directed edge from $v_{\tau_2 - 1}$
            to $v_{\tau_1}$ translates into saying that we have
            $h_{\tau_1} \notin \Phi(v_{\tau_2 - 1})$, as desired.


        \item
            We assume that $\Ev_1$ and $\Ev_2$ occur, and we can also assume
            $h_{\tau_1} \in \Phi_X(w^i)$, $h_{\tau_2} \in \Phi_Y(w^i)$,
            without loss of generality.
            We can also assume that the next milestone has not already been hit
            before the time $\tau_3$ (i.e., we assume $v_{\tau_3 - 1} \notin W^i$).
            Now we want to show that with probability at least $\frac{1}{5}$
            we have $v_{\tau_3} \in W^i$.

            Analogously to the proof of \emph{(ii)} we note that $h_{\tau_3}$ is
            chosen uniformly at random from a set of cardinality at most $5
            \cdot 2^{i-1}$, and
            $\Phi_X(w^i)$ is a subset of this set, with cardinality $2^{i-1}$.
            This shows already that, with probability at least $\frac{1}{5}$,
            $h_{\tau_3} \in \Phi_X(w^i)$.
            If we can now show that we have
            \[ v_{\tau_3 - 1} \cap \Phi_Y(w^i) \neq \emptyset \]
            then the claim will follow
            (because then, with probability at least $\frac{1}{5}$,
            we have $v_{\tau_3} \subseteq {\Phi(w^i)}$ and hence
            $v_{\tau_3} \in W^i$,
            cf.~\cref{Lemma: list}(\ref{Item: constraint containment})).

            By \cref{Lemma: adjacency} we know that $v_{\tau_2}$ and
            $v_{\tau_3 - 1}$ share either their $X$-coordinate or their
            $Y$-coordinate.
            If they share their $Y$-coordinate then we are done, because then
            $h_{\tau_2} \in v_{\tau_3 - 1} \cap \Phi_Y(w^i)$.
            Assume now that they share their $X$-coordinate but not their
            $Y$-coordinate;
            we will now examine the $X$- and $Y$-coordinates of
            $v_{\tau_3 - 1}$ separately, and show that neither of them can be
            an element of the set $\Phi(w^{i+1})$,
            which by lemma \cref{Lemma: list}(\ref{Item: constraint
            containment 2}) yields a contradiction to the fact that we have
            $v_{\tau_3 - 1} \in W^{i+1}$.
            \begin{itemize}
                \item[--]
                    The $Y$-coordinate of $v_{\tau_3-1}$ is given by $h_{\tau_3 -
                    1}$, which by definition of our stopping times is \emph{not} an
                    element of $\Phi(w^{i+1})$.
                \item[--]
                    Let $x$ denote the $X$-coordinate of $v_{\tau_3-1}$.
                    This $x$ is also the $X$-coordinate of $v_{\tau_2}$,
                    and also the one of $v_{\tau_2 - 1}$.
                    Suppose we had $x \in \Phi(w^{i+1})$:
                    Then by definition of our stopping times $x$ cannot be the
                    pivot $h_{\tau_2 - 1}$; hence, $v_{\tau_2 - 1} = \{ x,
                    h_{\tau_2 - 1} \}$ and $h_{\tau_2 - 1} \in Y$.
                    By \cref{Lemma: adjacency} we obtain that $\{ h_{\tau_1 + 1},
                    \dots, h_{\tau_2-1} \}$ is a subset of $Y$, so that as a
                    consequence $v_1$ and $v_{\tau_2 - 1}$ share their $X$-coordinate.
                    But the $X$-coordinate of $v_1$ is $h_{\tau_1}$; thus we
                    have found $x = h_{\tau_1} \in \Phi(w^i)$
                    and $v_{\tau_2} = \{ x, h_{\tau_2} \} \subseteq
                    \Phi(w^i)$.
                    By \cref{Lemma: list}(\ref{Item: constraint containment})
                    we obtain $v_{\tau_2} \in W^i$.
                    So the next milestone has already been hit by the time
                    $\tau_2$: a contradiction.
            \end{itemize}
    \end{enumerate}
\end{proof}

This concludes this section, establishing that the expected time going from
one milestone to the next is bounded by $O(\log n)$ (proposition
\ref{Proposition: next milestone}).
A very similar argumentation can be used to also bound by $O(\log n)$ the expected
time until the initial milestone is hit, yielding proposition
\ref{Proposition: first milestone}.
Theorem \ref{Theorem: upper bound} now follows by observing that there are
$O(\log n)$ many milestones.

%
%

\section{Conclusion}

In this note we have shown that the performance of \textsc{Random-Edge} on
simple $d$-polytopes with $d+2$ facets does not suffer if the improving
directions are specified by an arbitrary unique sink orientation.

The exact performance of \textsc{Random-Edge} on simple $d$-polytopes with
$d+k$ facets, where $k \ge 3$ is considered constant, remains an open problem.
The question is open even for $k=3$ and with the improving directions
specified by a linear objective function.

\section{Appendix}

The lemma and the theorem below, both of elementary nature, are used in
\cref{Section: proof}.

\begin{lemma}
   \label{Lemma: multiplicative bound}
   Let $X_1, X_2, \dots$ be non-negative random variables, let $N$ be a
   random variable with values in $\NN_0 \cup \{\infty\}$,
   and let $X = \sum_{j=1}^N X_j$.
   Assume that there is $M > 0$ such that for all $j$,
           $\E\bigl[ X_j \given N \ge j \bigr] \,\le\, M$.
   Then
   \[
       \E X \,\le\, M \cdot \E N.
   \]
\end{lemma}

\begin{proof}
    Without loss of generality we can assume that
    $\E\bigl[ X_j \given j > N \bigr]$ vanishes for all $j$.
    (If this is not the case then we can just replace each $X_j$ by the random
    variable that equals $X_j$ when $j \le N$, and equals $0$ otherwise.)
    Now we can write the expectation of each individual variable $X_j$ as
    \[
        \E X_j \,=\, \underbrace{ \E\bigl[ X_j \given N \ge j \bigr] }_{\le M} \,\Pr\bigl[ N \ge j \bigr]
        \,+\, \underbrace{ \E\bigl[ X_j \given N < j \bigr] }_{=0} \,\Pr \bigl[ N < j \bigr]
        \,\le\, M \Pr\bigl[ N \ge j \bigr]
    \]
    and conclude using monotone convergence:
    \begin{align*}
        \E X ~=~      \sum_{j \ge 1} \E X_j
             ~\le~    \sum_{j \ge 1} M \cdot \Pr\bigl[ N \ge j \bigr]
              ~=~      M \cdot \E N .
    \end{align*}
\end{proof}

       \begin{theorem}
          \label{Theorem: Bernoulli process}
          For $p \in (0,1)$ and $n \in \NN_0$, let $\tau(p,n)$ denote the first time
          that, during a Bernoulli process with parameter $p$, one has encountered
          $n$ \emph{successive} successes.
          Its expectation is given by
          \[
              \E \bigl[ \tau(p,n) \bigr] \,=\, \frac{1 - p^n}{p^n (1-p)}.
          \]
       \end{theorem}

\begin{proof}
    By induction on $n$.
    The equality clearly holds for $n=0$.
    For $n \ge 1$ we abbreviate the left-hand side by $t_n$ and find the
    recursion
    \[
        t_n \,=\, t_{n-1} + p + (1-p)(1 + t_n),
    \]
    which gives
    \begin{align*}
        t_n \,&=\,
        \frac{1 + t_{n-1}}{p}
        \,=\,
        \frac{1}{p} \left(
            1 + \frac{1-p^{n-1}}{p^{n-1}(1-p)}
        \right)
        \,=\,
        \frac{1-p^n}{p^n(1-p)}.
    \end{align*}
\end{proof}

\bibliographystyle{plainnat}
\bibliography{arxivpaper}

\end{document}